\documentclass[a4paper]{article}

\usepackage{amsthm}
\theoremstyle{plain}
\newtheorem{theorem}{Theorem}

\theoremstyle{definition}
\newtheorem{definition}{Definition}
\theoremstyle{remark}
\newtheorem{remark}{Remark}

\usepackage[sorting=ydnt,backend=bibtex,sortcites=true]{biblatex}
\usepackage{dmbiblatex}

\bibliography{article_Monniaux_computability_polyhedral_abstraction}

\usepackage{graphicx}
\usepackage{paralist}
\usepackage[utf8]{inputenc}
\usepackage{amsmath,amsfonts}

\newcommand{\ZZ}{\mathbb{Z}}
\newcommand{\QQ}{\mathbb{Q}}
\newcommand{\RR}{\mathbb{R}}

\newcommand{\tuple}[1]{\langle #1 \rangle}
\newcommand{\defn}{\stackrel{\triangle}{=}}
\newcommand{\calD}{\mathcal{D}}
\newcommand{\calT}{\mathcal{T}}
\newcommand{\parts}[1]{\mathcal{P}\left(#1\right)}
\newcommand{\vx}{\mathbf{x}}

\usepackage{hyperref}

\begin{document}

\title{On the decidability of the existence of polyhedral invariants in transition systems%
\thanks{This work was partially supported by the \href{http://erc.europa.eu/}{European Research Council} under the European Union's Seventh Framework Programme (FP/2007-2013) / ERC Grant Agreement nr. 306595 \href{http://stator.imag.fr}{``STATOR''.}}%
}

\author{David Monniaux\\
Univ. Grenoble Alpes, CNRS, Grenoble INP%
\thanks{Institute of Engineering Univ. Grenoble Alpes},
           VERIMAG, 38000 Grenoble\\
              \url{http://www-verimag.imag.fr/~monniaux/}}

\maketitle

\begin{abstract}
Automated program verification often proceeds by exhibiting inductive invariants entailing the desired properties.
For numerical properties, a classical class of invariants is convex polyhedra: solution sets of system of linear (in)equalities.
Forty years of research on convex polyhedral invariants have focused, on the one hand, on identifying ``easier'' subclasses, on the other hand on heuristics for finding general convex polyhedra.
These heuristics are however not guaranteed to find polyhedral inductive invariants when they exist.
To our best knowledge, the existence of polyhedral inductive invariants has never been proved to be undecidable.

In this article, we show that the existence of convex polyhedral invariants is undecidable, even if there is only one control state in addition to the ``bad'' one.
The question is still open if one is not allowed any nonlinear constraint.
\end{abstract}

\section{Introduction}
\label{sec:intro}
Methods for proving that a program can never enter some undesirable states generally rely on exhibiting an \emph{inductive invariant} $I$ --- a set of states containing the initial state(s) of the system, such that there exists no transition from $I$ to outside of $I$ --- containing no undesirable states.
If the program is safe, then the set of its reachable states is such an inductive invariant;
this set is however in general very complex and approaches based on \emph{abstract interpretation} \cite{DBLP:conf/popl/CousotC77} instead seek $I$ within a restricted class called an \emph{abstract domain}.
In the case of programs defined by a control flow graph $(V,E)$, $I$ is typically presented as a labeling of~$V$.
For instance, in \emph{interval analysis} \cite{DBLP:conf/popl/CousotC77}, each control location in is labeled either with the empty set, either with an interval for each numeric variable;
the inductiveness of such a labeling may be checked by interval arithmetic.

Interval analysis however cannot express relationships between variables.
In contrast, an example of \emph{relational analysis} is that of \emph{convex polyhedra}~\cite{Halbwachs_PhD,DBLP:conf/popl/CousotH78}: each control location is labeled with a set of linear (in)equalities over the program variables.
An important question is: given a transition system, and an undesirable control location (for instance, expressing that arithmetic overflow or out-of-bound access has occurred), is there any inductive invariant defined by labeling each control location with a convex polyhedron such that the undesirable location is labeled with the empty set, and thus proved to be inaccessible?

To our best knowledge, all the literature on the inference of polyhedral invariants has focused on two issues:
identifying ``simpler'' classes of convex polyhedra (the intervals form such a subclass), and heuristics for finding suitable (in)equalities.
Cousot \& Halbwachs' seminal work on polyhedral invariants \cite{Halbwachs_PhD,DBLP:conf/popl/CousotH78} proposed three such heuristics:
\begin{inparaenum}[i)]
\item propagating polyhedra along the edges (computing their image by the transition on the edge),
\item computing the convex hull of all incoming polyhedra if there are several incoming edges,
\item ``widening'' polyhedra by discarding unstable (in)equalities, ensuring termination of the analysis.
\end{inparaenum}
Further work included extracting potentially useful (in)equalities from the program, and less aggressive widening operators, in the hope of finding good invariants that coarser methods would have missed.

Yet none of these works considered showing that heuristics were needed in the first place, that is, showing that there is no general complete algorithm for finding inductive polyhedral invariants --- such an algorithm would, given a transition system and an undesirable state, either find an inductive invariant proving that this state is inaccessible, either find that such an invariant does not exist.

In the case of some restricted classes of polyhedra, the existence or nonexistence of a suitable invariant is decidable; this is most notably the case if a ``template'', a finite catalog of possible normal vectors for the constraints, is fixed in advance.
One may then study the complexity of this decision problem.
If the program variables are real (as opposed to integer) and the transitions on the edges are expressed by a first-order linear arithmetic formula with an existential quantifier prefix, the problem is $\Sigma^p_2$-complete~\cite{Gawlitza_Monniaux_LMCS12}.
We shall see that, even when templates are not used, it is sufficient to bound the number of constraints per polyhedron to make the problem decidable.
The issue is then what happens with an unbounded number of constraints.

In this article, we show that this problem is undecidable even for programs with one control state (plus one final ``bad state''), as long as one can use at least one nonlinear (polynomial) condition.

\section{Undecidability of the separating inductive polyhedral invariant existence problem}

\begin{definition}[Inductive separating invariant]
Let $n \geq 0$, let $\mathcal{D} \subseteq \parts{\QQ^n}$ be a class of properties, and let $T \subseteq \parts{\QQ^n \times \QQ^n}$ be a class of possible transition relations.

Let $C$ be a set of control states (including $c_i$ and $c_f$) and for all $j,k \in C$ let $\tau_{j,k} \subseteq \QQ^n \times \QQ^n$ be a transition relation defined by a formula in $T$; let $x_0 \in \QQ^n$. Decide whether there is an inductive invariant $I \in \calD^C$ such that $x_0 \in I(c_i)$ but $I(c_f) = \emptyset$ --- that is, $I$ such that
\begin{equation}\label{eqn:stability}
\forall j,k \in C ~ \forall x,x' \in \QQ^n ~ x \in I(j) \land (x,x')\in\tau_{j,k} \Rightarrow x' \in I(k)
\end{equation}
\end{definition}

Note that for $\calD = \parts{\QQ^n}$, this problem is equivalent to the negation of reachability of $c_f$ from $c_i$, thus, if $n \geq 2$ and $\calT$ includes first order linear arithmetic, this problem is undecidable, by simulation of a two-counter machine.

\begin{remark}
Consider the above definition with a real state $x \in \RR^n$.
If $\calD$ is the set of convex polyhedra over $\RR^n$ with at most $k$ constraints, and the transitions contain linear arithmetic updates and polynomial guards, the inductive separation problem is decidable.
\end{remark}

\begin{proof}
The stability condition (\ref{eqn:stability}) is then a formula in the theory of real closed fields, with free variables the coefficients of the $k$ polyhedral constraints.
We thus reduce the problem of the existence of the invariant to that of satisfiability of such a formula, which is decidable.
\end{proof}

\begin{theorem}
If $\calD$ is the set of convex polyhedra over $\QQ^n$, and the transitions contain (deterministic) linear arithmetic assignments and polynomial guards, the inductive separation problem is undecidable.
\end{theorem}

\begin{proof}
We reduce the halting problem for a deterministic machine over $m$ registers and $s$ states, with linear updates and linear (or polynomial) arithmetic constraints on the transitions, to the problem of existence of a convex polyhedral inductive separating invariant for a machine with $m+2$ registers and $s+1$ states, again with linear updates and polynomial guards on the transitions.

The target machine is obtained from the source as follows:
\begin{itemize}
\item the control states are the same, to which we add a special ``bad state''~$\sigma_b$
\item the first $m$ registers encode the registers of the source machine
\item another register, called $t$, is a time counter: initialized to $0$, it is incremented at every step
\item another register, called $y$, is initialized to $0$ and updated as $y := y+t+1$
\item the transitions are the same as those of the source machine, conjoined with a $y = (t^2+t)/2$ guard ($y \leq (t^2+t)/2$ also works), to which we add transitions from all states to $\sigma_b$ with the guard $y < (t^2+t)/2$.
\end{itemize}

\paragraph{If the source machine terminates, then there exists a bounded polyhedral inductive invariant}
Let $(\sigma_i,\vx_i)_{1 \leq i \leq n}$ be the run of the source machine, when $\sigma_i$ is the control state and $\vx_i$ the vector of registers at time~$i$.
The run of the target machine is thus $(\sigma_i,\vx_i,i,(i+i^2)/2)_{1 \leq i \leq n}$.

For each control state $\sigma$, consider the bounded polyhedron obtained as the convex hull of those $(\sigma, \vx_i,i,(i+i^2)/2)_{1 \leq i \leq n}$ from the run (this polyhedron is thus empty if and only if this control state is unreachable).
This polyhedron lies above the $y \geq (t^2+t)/2$ parabola;
its points are partitioned into, on the one hand, the $(\sigma_i,\vx_i,i,(i+i^2)/2)_{1 \leq i \leq n}$, on the other hand, points satisfying $y > (t^2+t)/2$ and thus excluded from the transitions by the $y = (t^2+t)/2$ guard.

The collection of these polyhedra thus forms an inductive invariant, with $\sigma_b$ labeled with the empty polyhedron.

\paragraph{If there exists a polyhedral separating invariant, then the source machine terminates}
If any point in the polyhedra lies below $y = (t^2+t)/2$, then $\sigma_b$ is reachable and thus the invariant is not separating;
and any point on $y = (t^2+t)/2$ must be a vertex.

Suppose an infinite run $(\sigma_i,\vx_i,i,(i+i^2)/2)_{1 \leq i}$, then there exists at least one of the control states $\sigma$ that repeats infinitely.
The corresponding $(\sigma_i,\vx_i,i,(i+i^2)/2)$ are vertices of the polyhedron associated to $\sigma$, but there is an infinity of them, which is absurd.\qed
\end{proof}

The same theorem and proof are valid for the set of bounded convex polyhedra over $\QQ^n$, the sets of bounded and unbounded convex polyhedra over $\ZZ^n$.

\section{Reduction to two control states}
We consider a transition system operating over $d \geq 1$ integer (resp. rational) variables and $N \geq 1$ control points. $0$ shall be the initial control point and $N-1$ is the final ``bad state''.

The state of the system consists in $\tuple{q,x_1,\dots,x_d}$ where $0 \leq q < N$ is a control state and $x_1,\dots,x_d$ are integer (resp. rational) variables.
The initial state is $\sigma_0 \defn \tuple{0, 0, \dots, 0}$.
The set of states $\Sigma$ is thus $\{ 0, \dots, N-1 \} \times \ZZ^d$ (resp. $\{ 0, \dots, N-1 \} \times \QQ^d$).
In the $\tau$ and $B$ formula, we request, without loss of generality, the only atomic propositions involving $q$ to be of the form $q = k$ where $0 \leq k < N$ is a constant.

\begin{theorem}
The inductive separating polyhedral invariant problem for $N$ control states in dimension $d$ reduces to the inductive separating polyhedral invariant problem for $2$ control states in dimension $d+(N-1)$.
\end{theorem}

\begin{proof}
  In the new system of transitions, a state consists in $\tuple{x_1, \dots, x_d, y_1, \dots, y_{N-1}}$.
  The variables $x_1, \dots, x_d$ encode the registers of the original system.
  The variables $y_1, \dots, y_{N-1}$ indicate which control state $q$ of the original system is simulated by encoding it into a vertex of a simplex: $q=0$ is encoded as $\tuple{0, \dots, 0}$, and $q > 0$ is encoded into $e_q = \tuple{\underbrace{0, \dots, 0}_{q-1 \text{~times}}, 1, 0, \dots, 0}$.

  The transition relation $\tau'$ the new system is obtained by replacing $q=k$ in each formula by $\tuple{y_1, \dots, y_{N-1}}=e_q$ where $e_q$ is the encoding of control state~$q$.

  The target system simulates concretely the source system, and conversely; the concrete reachability problems are equivalent.

Let us now see that any polyhedral inductive invariant within the source system (one polyhedron per control location) can be translated into a polyhedral invariant within the target system (one polyhedron, for the single control location), and the converse.

  A polyhedral inductive invariant within the source system consists in one polyhedron $P_q$ per control location $q$; it is translated into a polyhedron obtained as the convex hull of $\bigcup_q P_q \times \{ e_q \}$.
  Conversely, a polyhedral invariant $P$ in the target system is translated into the polyhedral invariant defined as $P_q = \{ \tuple{x_1,\dots,x_d} \mid (\tuple{x_1,\dots,x_d},e_q) \in P \}$. 
\end{proof}

\section{Perspectives}
Most studies of the computability and complexity issues in hardware and software verification have focused on the (concrete) reachability problem for various classes of transition systems.
Here, we instead consider the problem that abstract interpreters (or other classes of tools, for instance those based on counterexample-guided abstraction refinement) solve heuristically: finding an invariant within a given class; or, to make it a decision problem, deciding the existence of such an invariant.%
\footnote{For most abstract domains, it is possible to find a suitable invariant if one is known to exist, by enumerating all candidate invariants and stopping when an inductive separating one is found. This justifies the use of the existence decision problem for the purposes of theoretical computability studies.}

We started this work hoping to vindicate forty years of research on heuristics by showing that the existence of polyhedral inductive separating invariants in a system with transitions in linear arithmetic (integer or rational) is undecidable;
we had to settle for a weaker result (involving nonlinear arithmetic).
The original question is still open; we have discussed it with experts, to no avail.

\printbibliography
\end{document}